\documentclass{article}
\usepackage{amsmath,amssymb,array}
\usepackage{authblk}
\usepackage{amsthm}
\usepackage{verbatim}
\usepackage[dvips]{graphicx}
\usepackage{float}
\usepackage[dvips]{graphicx}
\usepackage[algoruled,linesnumbered,noend]{algorithm2e}
\usepackage[utf8]{inputenc}
\usepackage{tikz}
\usetikzlibrary{arrows,snakes,backgrounds}

\begin{document}

\title{On the Complexity of Minimum Labeling Alignment of Two Genomes}
\author[1]{Riccardo Dondi\thanks{riccardo.dondi@unibg.it}}
\author[2]{Nadia El-Mabrouk\thanks{mabrouk@iro.umontreal.ca}}
\affil[1]{Dipartimento di Scienze dei Linguaggi, della Comunicazione e degli Studi Culturali,
Universit\`a degli Studi di Bergamo, Bergamo - Italy}
\affil[2]{Département d’Informatique et Recherche Opérationnelle, Université de Montréal, Montréal - Canada}


\maketitle

\newcommand{\nota}[1]{{\sffamily\bfseries #1}\marginpar{\textbf{\framebox{X}}}}

\newcommand{\gx}{\ensuremath{\mathcal{X}}}
\newcommand{\gy}{\ensuremath{\mathcal{Y}}}

\newcommand{\tya}{\emph{type a labeling}}
\newcommand{\tyb}{\emph{type b labeling}}

\newcommand{\cluster}{\ensuremath{\mathcal{C}}}
\newcommand{\lchild}{\ensuremath{l}}
\newcommand{\rchild}{\ensuremath{r}}
\newcommand{\leaf}{\ensuremath{\Lambda}}
\newcommand{\leaveset}{\ensuremath{L}}
\newcommand{\radice}{\ensuremath{r}}
\newcommand{\lca}{\ensuremath{\text{lca}}}
\newcommand{\occ}{\ensuremath{\text{occ}}}

\newfloat{Algorithm}{h}{loa}
\newtheorem{Thm}{Theorem}[section]
\newtheorem{lemma}[Thm]{Lemma}
\newtheorem{Prop}[Thm]{Proposition}
\newtheorem{Obs}[Thm]{Observation}
\newtheorem{remark}[Thm]{Remark}
\newtheorem{Cor}[Thm]{Corollary}
\newtheorem{Cl}[Thm]{Claim}

\theoremstyle{definition}
\newtheorem{Def}{Definition}[section]
\newtheorem{PB}{Problem}
\newtheorem{Example}{Example}[section]

\begin{abstract}

In this note we investigate the complexity of the Minimum Label
Alignment problem and we show that such a problem is APX-hard.

\end{abstract}

\section{Introduction}
\label{sec:intro}

In this note we consider the computational (and approximation) 
complexity of the Minimum Label Alignment problem. This problem has been
recently introduced in bioinformatics to deal with
the inference of evolutionary scenarios for genome organization \cite{Recomb2012}.
In this note we show that the problem is APX-hard, even when the genome
contains at most five occurrences of the same gene.
The results implies that the Duplication-Loss Alignment problem 
and the Two Species Small Phylogeny problem introduced in \cite{Recomb2012}
are not in even in NP.

Next, we introduce some preliminary definitions.
A genome is considered as a string over alphabet $\Sigma$.
The i-th character of a genomes $\gx$ is denoted by $\gx_i$.
Two aligned genomes $\gx$, $\gy$ are two aligned strings 
 over alphabet $\Sigma^- = \Sigma \cup \{ - \}$
(where $-$ denotes a gap in the alignment) such that
if $\gx_i \neq -$ and $\gy_i \neq -$, then $\gx_i = \gy_i$
and $\gx_i$, $\gy_i$ cannot be both equal to $-$.
Two aligned genomes can be seen as a matrix of size
$2 \times m$ (where $m$ is the size of the alignment).
A column is a \emph{match} if it does not contain a gap. 
A \emph{labeling} of an aligned genome $\gx$ is an
interpretation of the unmatched characters of $\gx$ in terms of duplications
and losses. 
A duplication can be represented as a directed arc from a substring
of $\gx$ to a different identical substring of $\gx$.
A labeling is \emph{feasible} if it induces no cycle.
Consider a duplication in $\gx$ from a substring $s$ to a substring $t$.
Such a duplication is called \emph{maximal} if $s$ and $t$ are two identical
maximal substrings in $\gx$, that is if the characters on the left of $s$ and $t$ in $\gx$ 
are different (or one of these characters does not exist) and if the characters on the right of
 $s$ and $t$ in $\gx$
are different (or one of these characters does not exist).

Giving a cost function $c$ that defines the cost of the possible operations considered (duplications and losses), the cost of a labeling of $\gx$, $\gy$ is the sum
of the costs of the underlying operations.

We investigate the complexity of the following problem.
\begin{PB}
\label{Problem:MLA} \textbf{Minimum Labeling Alignment}[MLA]\\
\textbf{Input}: two aligned genomes $\gx$ and $\gy$.\\
\textbf{Output}: a minimum cost feasible labeling $L$ of $\gx$ and $\gy$.
\end{PB}

In what follows, given a graph $G=(V,E)$ and a vertex $v \in V$, 
we denote by $N(v)$ the set of vertices adjacent to $v$ in $G$.
A graph $G=(V,E)$ is cubic when $N(v)=3$ for each $v \in V$.

\section{Hardness of Minimum Labeling Alignment}
\label{section:MLA-APX-hard}

We prove that the MLA problem is APX-hard, even if each symbol (gene)
has at most $5$ occurrences in $\gx$ or $\gy$, by giving a reduction 
(more precisely an $L$-reduction \cite{Ausiello-Crescenzi}) from the Minimum Vertex Cover problem 
on Cubic graphs (MVCC) to MLA. Notice that MVCC is known to be APX-hard \cite{AK92}.


\begin{PB}
\label{Problem:MinVC} \textbf{Minimum Vertex Cover Problem on Cubic graph} [MVCC]\\
\textbf{Input}: a cubic graph $G=(V,E)$, where $V = \{ v_1, \dots, v_n \}$.\\
\textbf{Output}: a minimum cardinality set $V' \subseteq V$, 
such that for each $\{ v_i, v_j \} \in E$, 
at least one of $v_i$, $v_j$ belongs to $V'$.
\end{PB}

Next, we present the L-reduction from MVCC to MLA.
Let $G=(V,E)$ be a cubic graph. Define the following ordering on the
edges in $E$: $\{ v_i,v_j \} < \{ v_x, v_y \}$ if and only if
$i < x$, or (in case $i = x$) $ j < y$.
We denote by $\{v_1, v_a\}$
and $\{v_z, v_w\}$ the first and the last edges of $E$.
Notice that, based on this ordering, we denote the edges incident on $v_i$, as the first, 
the second and the third edges of $v_i$.
Furthermore, in what follows, given a vertex $v_i \in V$, we denote with
$\{ v_i,v_j \}$, $\{ v_i,v_h \}$, $\{ v_i,v_k \}$ the three
edges of $G$ incident on $v_i$.

Now, we define the corresponding 
aligned genomes $\gx$ and $\gy$ as follows.
First, we present an overview of the construction of $\gx$ and $\gy$.
The aligned genomes $\gx$ and $\gy$ consists of two \emph{parts} 
and each part is then divided
into \emph{blocks} (that is substrings): 
the leftmost part is called the \emph{Vertex-Edge-Set Part} (VE-Part), 
the rightmost part is called the \emph{Auxiliary Part} (A-Part) 
(see Fig. \ref{fig:struct}).

In the VE-part each position of $\gx$ is different from $-$, 
while $\gy$ contains some gaps. 
Each position of $\gx$ and $\gy$ in the A-part is a match, hence
$\gx$ and $\gy$ are identical in the A-part. 
By construction each position of the aligned genome $\gy$ 
is either a gap or it is a match, 
hence the characters of $\gy$ do not need any labeling. 
It follows that the definition of a labeling of $\gx$ and
$\gy$ is computed by labeling the unmatched elements in the VE-part of $\gx$.

The VE-part of $\gx$ and $\gy$ consists of the concatenation 
of $|V|+|E|$ blocks (see Fig. \ref{fig:struct}):
one block $B_{\gx-VE}(e_{i,j})$ ($B_{\gy-VE}(e_{i,j})$ respectively) in $\gx$
($\gy$ respectively) for each edge $\{ v_i, v_j \} \in E$, 
one block $B_{\gx-VE}(v_i)$ ($B_{\gy-VE}(v_i)$ respectively) in $\gx$
(in $\gy$ respectively) for each vertex $v_i \in V$. 

The A-part of $\gx$ and $\gy$ consists of the concatenation of $2|V|$ blocks (see Fig. \ref{fig:struct}): 
two blocks $B_{\gx,A,1}(v_i)$, $B_{\gx,A,2}(v_i)$ 
($B_{\gy,A,1}(v_i)$, $B_{\gy,A,2}(v_i)$ respectively) in  $\gx$ 
(in $\gy$ respectively), for each $v_i \in V$.

\begin{figure}
\[
\gx = \underbrace{B_{\gx-VE}(v_1) \dots B_{\gx-VE}(v_n) B_{\gx-VE}(e_{1,a}) \dots 
B_{\gx-VE}(e_{z,w})}_{\text{VE-part}} \cdot
\]
\[
\cdot \underbrace{B_{\gx,A,1}(v_1) B_{\gx,A,2}(v_1)
\dots B_{\gx,A,1}(v_n) B_{\gx,A,2}(v_n)}_{\text{A-part}}
\]

\[
\gy = \underbrace{B_{\gy-VE}(v_1) \dots B_{\gy-VE}(v_n) B_{\gy-VE}(e_{1,a}) \dots 
B_{\gy-VE}(e_{z,w})}_{\text{VE-part}} \cdot
\]
\[
\cdot \underbrace{B_{\gy,A,1}(v_1) B_{\gy,A,2}(v_1)
\dots B_{\gy,A,1}(v_n) B_{\gy,A,2}(v_n)}_{\text{A-part}}
\]

\caption{The structure of $\gx$ and $\gy$.}
\label{fig:struct}
\end{figure}

Now, we define the specific values of the blocks of $\gx$ and $\gy$.
Given an edge $\{ v_i, v_j \} \in E$, 
where $i < j$ and $\{ v_i, v_j \}$ is the $p$-th edge of $v_i$ 
and the $q$-th of $v_j$, $1 \leq p \leq 3$ and $1 \leq q \leq 3$, 
we define its associated blocks  $B_{\gx-VE}(e_{i,j})$,
$B_{\gy-VE}(e_{i,j})$. 
The block $B_{\gx-VE}(e_{i,j})$ is defined as follows:
\[
B_{\gx-VE}(e_{i,j}) = s_{e,i,j} x_{i,p} e_{i,j,1} e_{i,j,2} x_{j,q}
\]
The block $B_{\gy-VE}(e_{i,j})$ is defined as follows:
\[
B_{\gy-VE}(e_{i,j}) = s_{e,i,j} (-)^{4}
\]
Hence notice that $B_{\gx-VE}(e_{i,j})$ contains $4$ unmatched characters,
that is the characters $x_{i,p}$, $e_{i,j,1}$, $e_{i,j,2}$, $x_{j,q}$.

Now, we define the block $B_{\gx-VE}(v_i)$, with $v_i \in V$.
The $i$-encoding of $\{v_i, v_j \}$, $i-enc_{i,j}$, is defined as follows:
\begin{itemize}
\item $i-enc_{i,j} = x_{i,p} e_{i,j,1} e_{i,j,2}$
\end{itemize}
and $i-enc_{i,j}[l]=x_{i,p}$, $i-enc_{i,j}[r]= e_{i,j,1} e_{i,j,2}$.
The $j$-encoding of $\{v_i, v_j \}$, $j-enc_{i,j}$, is defined as follows:
\begin{itemize}
\item $j-enc_{i,j} = e_{i,j,1} e_{i,j,2} x_{j,q}$
\end{itemize}
and $j-enc_{i,j}[l]= e_{i,j,1} e_{i,j,2}$, $j-enc_{i,j}[r]= x_{j,q}$.

The block $B_{\gx-VE}(v_i)$ is defined as follows:
\[
B_{\gx-VE}(v_i) = s_i z_{i,1} z_{i,2}~i-enc_{i,j}~z_{i,3} z_{i,4}~i-enc_{i,h}~z_{i,5} z_{i,6}~i-enc_{i,k} ~z_{i,7} z_{i,8}
\]

The block $B_{\gy-VE}(v_i)$ is defined as follows:
\[
B_{\gy-VE}(v_i) = s_i (-)^{17}
\]
Hence notice that $B_{\gx-VE}(v_i)$ contains $17$ unmatched characters,
that is the substring 
$z_{i,1} z_{i,2}~i-enc_{i,j}~z_{i,3} z_{i,4}~i-enc_{i,h}~z_{i,5} z_{i,6}~i-enc_{i,k}~z_{i,7} z_{i,8}$.

Now, we define the A-part of $\gx$ and $\gy$. Recall that $\gx$ and
$\gy$ are identical in the A-part.
The block $B_{\gx,A,1}(v_i)$ is defined as follows:
\[
B_{\gx,A,1}(v_i) = w_{i,1} z_{i,1} z_{i,2} w_{i,2} z_{i,3} z_{i,4}
w_{i,3} z_{i,5} z_{i,6} w_{i,4} z_{i,7} z_{i,8}
\]
The block $B_{\gy,A,1}(v_i)$ is identical to $B_{\gx,A,1}(v_i)$.

The block $B_{\gx,A,2}(v_i)$ is defined as follows:
\[
B_{\gx,A,2}(v_i) = u_{i,1} z_{i,2}~i-enc_{i,j}[l]~u_{i,2}~i-enc_{i,j}[r]~z_{i,3} 
u_{i,3} z_{i,4}~i-enc_{i,h}[l]~u_{i,4}~i-enc_{i,h}[r]~z_{i,5} \cdot
\]
\[
\cdot u_{i,5} z_{i,6}~i-enc_{i,k}[l]~u_{i,6}~i-enc_{i,k}[r]~z_{i,7}
\]
The block $B_{\gy,A,2}(v_i)$ is identical to $B_{\gx,A,2}(v_i)$.

\begin{Example}
\label{Ex1}
A cubic graph $G=(V,E)$ and the the corresponding genome $\gx$.

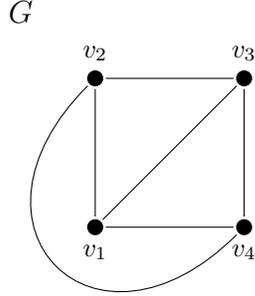
\begin{tikzpicture}[scale=.99,shorten <=1pt,shorten >=1pt]
    \tikzstyle{vertex} = [circle,draw,fill=black,minimum
    size=5pt,inner sep=2pt]
    \tikzstyle{edge} = [draw,-]
 \node at (-1,2.9) {\large{$G$}};
 \node[vertex] (v1) at (0,0) [label=below:$v_1$] {};
 \node[vertex] (v2) at (0,2) [label=above:$v_2$] {};
 \node[vertex] (v3) at (2,2) [label=above:$v_3$] {};
 \node[vertex] (v4) at (2,0) [label=below:$v_4$] {};
 
  \node (vs) at (4,0) {};

   \foreach \source/ \dest in{v1/v2, v1/v3, v2/v3, v3/v4, v1/v4}
   \path[edge] (\source) -- (\dest);%
   \draw (v2) .. controls (-2,0) and (0,-2) .. (v4);
\end{tikzpicture}

First, we define the blocks $B_{X-VE}(e_{i,j})$ associated with edges 
$\{ v_i,v_j \} \in E$

\begin{itemize}
\item $B_{\gx-VE}(e_{1,2}) = s_{e,1,2} x_{1,1} e_{1,2,1} e_{1,2,2} x_{2,1}$

\item $B_{\gx-VE}(e_{1,3}) = s_{e,1,3} x_{1,2} e_{1,3,1} e_{1,3,2} x_{3,1}$

\item $B_{\gx-VE}(e_{1,4}) = s_{e,1,4} x_{1,3} e_{1,4,1} e_{1,4,2} x_{4,1}$

\item $B_{\gx-VE}(e_{2,3}) = s_{e,2,3} x_{2,2} e_{2,3,1} e_{2,3,2} x_{3,2}$

\item $B_{\gx-VE}(e_{2,4}) = s_{e,2,4} x_{2,3} e_{2,4,1} e_{2,4,2} x_{4,2}$

\item $B_{\gx-VE}(e_{3,4}) = s_{e,3,4} x_{3,3} e_{3,4,1} e_{3,4,2} x_{4,3}$

\end{itemize}

Now, in order to define the block $B_{\gx-VE}(v_i)$, with $v_i \in V$, we have to
define the encoding of $i-enc_{i,j}$, $j-enc_{i,j}$:

\begin{itemize}

\item $1-enc_{1,2} = x_{1,1} e_{1,2,1} e_{1,2,2}$; $2-enc(1,2) = e_{1,2,1} e_{1,2,2} x_{1,2}$

\item $1-enc_{1,3} = x_{1,2} e_{1,3,1} e_{1,3,2}$; $3-enc(1,3) = e_{1,3,1} e_{1,3,2} x_{3,1}$

\item $1-enc_{1,4} = x_{1,3} e_{1,4,1} e_{1,4,2}$; $4-enc(1,4) = e_{1,4,1} e_{1,4,2} x_{4,1}$

\item $2-enc_{2,3} = x_{2,2} e_{2,3,1} e_{2,3,2}$; $3-enc(2,3) = e_{2,3,1} e_{2,3,2} x_{3,2}$

\item $2-enc_{2,4} = x_{2,3} e_{2,4,1} e_{2,4,2}$; $4-enc(2,4) = e_{2,4,1} e_{2,4,2} x_{4,2}$

\item $3-enc_{3,4} = x_{3,3} e_{3,4,1} e_{3,4,2}$; $4-enc(3,4) = e_{3,4,1} e_{3,4,2} x_{4,3}$
\end{itemize}
\noindent
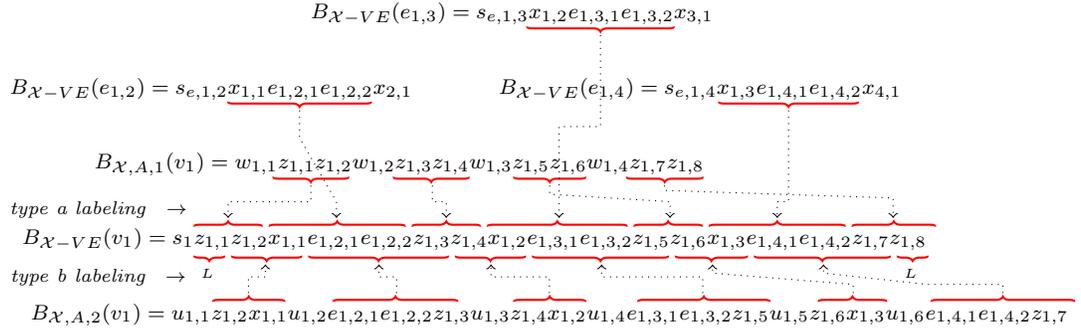
\begin{tikzpicture}[scale=.5,shorten <=1pt,shorten >=1pt]
    \tikzstyle{vertex} = [circle,draw,fill=black,minimum
    size=5pt,inner sep=2pt]
    \tikzstyle{edge} = [draw,-]

\node at (-10,5) {\footnotesize{$B_{\gx-VE}(e_{1,3}) = s_{e,1,3} x_{1,2} e_{1,3,1} e_{1,3,2} x_{3,1}$}};  

\node at (-18,3) {\footnotesize{$B_{\gx-VE}(e_{1,2}) = s_{e,1,2} x_{1,1} e_{1,2,1} e_{1,2,2} x_{2,1}$}};    
    
\node at (-5, 3) {\footnotesize{$B_{\gx-VE}(e_{1,4}) = s_{e,1,4} x_{1,3} e_{1,4,1} e_{1,4,2} x_{4,1}$}};   

\node at (-13, 1) {\footnotesize{$B_{\gx,A,1}(v_1) = w_{1,1} z_{1,1} z_{1,2} w_{1,2} z_{1,3} z_{1,4}
w_{1,3} z_{1,5} z_{1,6} w_{1,4} z_{1,7} z_{1,8}$}};   
     
\node at (-21,-0.2){ \scriptsize{$\tya~~ \rightarrow$}};
\node at (-11,-1) {\footnotesize{$B_{\gx-VE}(v_1) = s_1 z_{1,1} z_{1,2} x_{1,1} e_{1,2,1} e_{1,2,2} 
z_{1,3} z_{1,4} x_{1,2} e_{1,3,1} e_{1,3,2} z_{1,5} z_{1,6} 
x_{1,3} e_{1,4,1} e_{1,4,2} z_{1,7} z_{1,8}$}};

\node at (-21,-2){ \scriptsize{$\tyb~~ \rightarrow$}};

\node at (-9,-3) {\footnotesize{$B_{\gx,A,2}(v_1) = 
u_{1,1} z_{1,2} x_{1,1} u_{1,2} e_{1,2,1} e_{1,2,2} 
z_{1,3} u_{1,3} z_{1,4} x_{1,2} u_{1,4} e_{1,3,1} e_{1,3,2} z_{1,5} u_{1,5}
z_{1,6} x_{1,3} u_{1,6} 
e_{1,4,1} e_{1,4,2} z_{1,7}$}};

\draw[snake=brace,red,thick] (-18.5,-.7)--(-16.6,-.7);
\draw[snake=brace,red,thick] (-16.5,-.7)--(-12.8,-.7);
\draw[snake=brace,red,thick] (-12.7,-.7)--(-10.8,-.7);
\draw[snake=brace,red,thick] (-10.7,-.7)--(-6.9,-.7);
\draw[snake=brace,red,thick] (-6.8,-.7)--(-4.9,-.7);
\draw[snake=brace,red,thick] (-4.8,-.7)--(-1.1,-.7);
\draw[snake=brace,red,thick] (-1,-.7)--(1.3,-.7);

\draw[snake=brace,mirror snake,red,thick] (-9.65,4.7)--(-5.65,4.7);
\draw[snake=brace,mirror snake,red,thick] (-17.6,2.7)--(-13.7,2.7);
\draw[snake=brace,mirror snake,red,thick] (-4.55,2.7)--(-0.75,2.7);
\draw[snake=brace,mirror snake,red,thick] (-16.4,0.7)--(-14.3,0.7);
\draw[snake=brace,mirror snake,red,thick] (-13.2,0.7)--(-11.1,0.7);
\draw[snake=brace,mirror snake,red,thick] (-10,0.7)--(-8,0.7);
\draw[snake=brace,mirror snake,red,thick] (-7,0.7)--(-4.88,0.7);

\draw[snake=brace,mirror snake,red,thick] (-18.5,-1.4)--(-17.6,-1.4);
\node at (-18.1,-1.9) {\tiny{$L$}}; 
\draw[snake=brace,mirror snake,red,thick] (-17.5,-1.4)--(-15.55,-1.4);
\draw[snake=brace,mirror snake,red,thick] (-15.45,-1.4)--(-11.65,-1.4);
\draw[snake=brace,mirror snake,red,thick] (-11.55,-1.4)--(-9.6,-1.4);
\draw[snake=brace,mirror snake,red,thick] (-9.5,-1.4)--(-5.8,-1.4);
\draw[snake=brace,mirror snake,red,thick] (-5.7,-1.4)--(-3.7,-1.4);
\draw[snake=brace,mirror snake,red,thick] (-3.6,-1.4)--(0.1,-1.4);
\draw[snake=brace,mirror snake,red,thick] (0.2,-1.4)--(1.1,-1.4);
\node at (0.6,-1.9) {\tiny{$L$}}; 

\draw[snake=brace,red,thick] (-18,-2.7)--(-16,-2.7);
\draw[snake=brace,red,thick] (-14.8,-2.7)--(-11.4,-2.7);
\draw[snake=brace,red,thick] (-10,-2.7)--(-8,-2.7);
\draw[snake=brace,red,thick] (-6.7,-2.7)--(-3.1,-2.7);
\draw[snake=brace,red,thick] (-1.9,-2.7)--(0,-2.7);
\draw[snake=brace,red,thick] (1.1,-2.7)--(5,-2.7);

\draw [->,dotted] (-7.65,4.5) -- (-7.65,2) -- (-8.75,2) -- (-8.75,-0.5);  

\draw [->,dotted] (-15.65,2.5) -- (-15.65,1.7) -- (-14.67,0) -- (-14.67,-0.5);  

\draw [->,dotted] (-2.65,2.5) -- (-2.65,0) -- (-2.95,0) -- (-2.95,-0.5);  

\draw [->,dotted] (-15.35,.5) -- (-15.35,0) -- (-17.55,0) -- (-17.55,-0.5);  

\draw [->,dotted] (-12.12,.5) -- (-12.12,0) -- (-11.75,0) -- (-11.75,-0.5);  

\draw [->,dotted] (-9,.5) -- (-9,0.2) -- (-5.8,0) -- (-5.8,-0.5);  

\draw [->,dotted] (-5.95,.5) -- (-5.95,0.3) -- (0.13,0.1) -- (0.13,-0.5);  


\draw [->,dotted] (-17.01,-2.5) -- (-17.01,-2) -- (-16.56,-2) -- (-16.56,-1.6);  

\draw [->,dotted] (-13.07,-2.5) -- (-13.07,-2) -- (-13.54,-2) -- (-13.54,-1.6);  

\draw [->,dotted] (-9.01,-2.5) -- (-9.01,-2) -- (-10.56,-2) -- (-10.56,-1.6);  

\draw [->,dotted] (-4.93,-2.5) -- (-4.93,-2) -- (-7.62,-2) -- (-7.62,-1.6);  

\draw [->,dotted] (-0.93,-2.5) -- (-0.93,-2.2) -- (-4.67,-1.9) -- (-4.67,-1.6);  

\draw [->,dotted] (3.05,-2.5) -- (3.05,-2.2) -- (1.05,-2.2) -- (-1.72,-1.9) -- (-1.72,-1.6);  

\end{tikzpicture}

\vspace*{.5cm}

A $\tya$ for $B_{\gx-VE}(v_1)$ (in the upper part) and a $\tyb$ for $B_{\gx-VE}(v_1)$ (in the lower part).

\qed
\end{Example}

Now, we define the cost $c$ of labeling the aligned genome $\gx$ (recall that
$\gy$ does not need any labeling). 
Given an integer $z > 1$, then  the cost of a duplication of length $z$ is  
$c(D(z)) = 1$, while the cost of a loss of length $z$ is $c(L(z)) = z$.

Before giving the details of the proof, we give a high-level description of the reduction.
We will show that each block $B_{\gx-VE}(v_i)$ can be labeled essentially
in two possible ways (see Remark \ref{Rem:SolV} and Example \ref{Ex1}): 
\begin{enumerate}
\item with a $\tya$, defining maximal duplications 
from $B_{\gx-VE}(e_{i,j})$, $B_{\gx-VE}(e_{i,h})$, $B_{\gx-VE}(e_{i,k})$, $B_{\gx,A,1}(v_i)$
to $B_{\gx-VE}(v_i)$; a $\tya$ is the optimal labeling of $B_{\gx-VE}(v_i)$ 
(see Lemma \ref{lem:sol-loc-vertex}); 
\item with a $\tyb$,  defining maximal duplications in $B_{\gx-VE}(v_i)$
from the block $B_{\gx,A,2}(v_i)$ to  $B_{\gx-VE}(v_i)$; a $\tyb $ is a suboptimal
labeling of $B_{\gx-VE}(v_i)$ (see Lemma \ref{lem:sol-loc-vertex}).
\end{enumerate}
Thanks to the property of block $B_{\gx-VE}(e_{i,j})$ 
(see Remark \ref{Rem:SolE} and Lemma \ref{lem:sol-loc-edges}),
we will able to relate these two type of labelings with a cover of $G$ (see Lemma \ref{lem:hardness1}
and Lemma \ref{lem:hardness2}): a $\tyb$ for $B_{\gx-VE}(v_i)$ corresponds to a vertex $v_i$ 
in a vertex cover $V'$ of $G$, a $\tya$ for $B_{\gx-VE}(v_i)$ corresponds to a vertex $v_i$ 
in $V \setminus V'$ of $G$.

Now, we give the details of the reduction. First, we introduce 
some preliminaries properties of $\gx$ and $\gy$.

\begin{remark}
\label{Rem:SolV}
Given a cubic graph $G=(V,E)$, let $v_i$ be a vertex of $V$
such that $\{ v_i , v_j \}$, $\{ v_i , v_h \}$, $\{ v_i , v_k \}$
are the first, the second and the third edges of $v_i$ respectively.
Let $(\gx,\gy)$ be the corresponding instance of MLA.
The following labeling of $B_{\gx-VE}(v_i)$ 
(denoted as a $\tya$ for
$B_{\gx-VE}(v_i)$) has a cost of $7$ (it consists of $7$ duplications): 
\begin{itemize}
\item four duplications coming from the block $B_{\gx,A,1}(v_i)$, 
for the substrings $z_{i,2p-1}, z_{i,2p}$, $1 \leq p \leq 4$;
\item three duplications coming from the blocks 
$B_{\gx-VE}(e_{ij})$ (for the substring $i-enc_{i,j}$), 
$B_{\gx-VE}(e_{ih})$ (for the substring $i-enc_{i,h}$), $B_{\gx-VE}(e_{ik})$ 
(for the substring $i-enc_{i,k}$).
\end{itemize}
The following labeling of $B_{\gx-VE}(v_i)$ 
(denoted as a $\tyb$ for $B_{\gx-VE}(v_i)$) has a cost of $8$ (it consists of $6$ duplications and $2$ losses): 
\begin{itemize}
\item six duplications from 
$B_{\gx,A,2}(v_i)$ to $B_{\gx-VE}(v_i)$  
(substrings $z_{i,2}~i-enc_{i,j}[l]$, $i-enc_{i,j}[r]~z_{i,3}$,
$z_{i,4}~i-enc_{i,h}[l]$, $i-enc_{i,h}[r]~z_{i,5}$, 
$z_{i,6}~i-enc_{i,k}[l]$, $i-enc_{i,k}[r]~z_{i,7})$;
\item two losses for the two substrings $z_{i,1}$ and $z_{i,8}$.
\end{itemize}
Notice that in a $\tyb$ for $B_{\gx-VE}(v_i)$, there is no
duplication of $B_{\gx-VE}(v_i)$ from 
substrings of $B_{\gx-VE}(e_{ij})$, $B_{\gx-VE}(e_{ih})$, $B_{\gx-VE}(e_{ik})$.
\end{remark}

\begin{remark}
\label{Rem:SolE}
Let $G=(V,E)$ be a cubic graph, let $\{v_i, v_j \} \in E$, with $i < j$, be the
$p$-th edge of $v_i$, $1 \leq p \leq 3$, and the $q$-th edge of $v_j$, $1 \leq q \leq 3$.
Let $(\gx,\gy)$ be the corresponding instance of MLA.
The following labeling of $B_{\gx-VE}(e_{i,j})$ has cost $2$: 
\begin{itemize}
\item one duplication coming  either from $B_{\gx-VE}(v_i)$ (for the substring $x_{i,p} e_{i,j,1} e_{i,j,2}$) 
or from $B_{\gx-VE}(v_j)$ (for the substring  $e_{i,j,1} e_{i,j,2} x_{j,q}$);
\item one loss either for the last character of $B_{\gx-VE}(e_{i,j})$ or for 
the second character of $B_{\gx-VE}(e_{i,j})$ 
(that is the unmatched character of $B_{\gx-VE}(v_j)$ not involved in the duplication).
\end{itemize}
\end{remark}

Now, we are ready to show that a $\tya$ is the only optimal labeling 
for $B_{\gx-VE}(v_j)$.

\begin{lemma}
\label{lem:sol-loc-vertex}
Let $G=(V,E)$ be an instance of MVCC and let $(\gx,\gy)$ be the 
corresponding instance of MLA. Then, given a block $B_{\gx-VE}(v_i)$, with $v_i \in V$:
(1) any feasible  labeling of  $B_{\gx-VE}(v_i)$ has a
cost of at least $7$; 
(2) if a labeling  has cost of $7$, then such a labeling
is a $\tya$ for $B_{\gx-VE}(v_i)$.
\end{lemma}
\begin{proof}
The proof that any feasible labeling of  $B_{\gx-VE}(v_i)$ needs a
cost of at least $7$ follows from a simple counting argument.
Notice that $B_{\gx-VE}(v_i)$ contains $17$ unmatched characters and that
$B_{\gx-VE}(v_i)$ is labeled by duplications of length at most $3$.
By construction, any feasible labeling of $B_{\gx-VE}(v_i)$ can define a duplication
of length at most $2$ that contains the leftmost unmatched character
of $B_{\gx-VE}(v_i)$. The same property holds for the rightmost unmatched character of $B_{\gx-VE}(v_i)$.
Hence, consider the unmatched characters of $B_{\gx-VE}(v_i)$ 
not labeled by one of these two labelings of the rightmost and leftmost characters 
of $B_{\gx-VE}(v_i)$.
Those characters of $B_{\gx-VE}(v_i)$ are at least $13$, 
and since each duplication has length at most $3$, it follows that
at least $\lceil \frac{13}{3} \rceil = 5$ duplications are required for  labeling 
these $13$ unmatched characters of $B_{\gx-VE}(v_i)$. This implies an overall cost of at least
$7$ for any labeling of $B_{\gx-VE}(v_i)$.

Now, we prove that if a feasible labeling  of $B_{\gx-VE}(v_i)$ has a cost
of $7$, then such a feasible labeling
must be a $\tya$ of $B_{\gx-VE}(v_i)$.
First, notice that if a labeling of $B_{\gx-VE}(v_i)$ contains only
duplications from  $B_{\gx-VE}(e_{i,j})$, $B_{\gx-VE}(e_{i,h})$,
$B_{\gx-VE}(e_{i,k})$, $B_{\gx,A,1}(v_i)$, then it has a cost of $7$ if and only if
is a $\tya$. Indeed, a $\tya$ is the only labeling that consists only 
of maximal duplications from $B_{\gx-VE}(e_{i,j})$, $B_{\gx-VE}(e_{i,h})$,
$B_{\gx-VE}(e_{i,k})$, $B_{\gx,A,1}(v_i)$ to $B_{\gx-VE}(v_i)$.

Now, assume that a labeling of $B_{\gx-VE}(v_i)$ contains only
duplications from  $B_{\gx,A,2}(v_i)$. 
A $\tyb$ is the only labeling of $B_{\gx-VE}(v_i)$ that 
consists only of maximal duplications from  
$B_{\gx,A,2}(v_i)$, hence
any other labeling of $B_{\gx-VE}(v_i)$ that contains only
duplications from  $B_{\gx,A,2}(v_i)$ requires a cost of at least $8$.

Hence, assume that a labeling $L$ of $B_{\gx-VE}(v_i)$ contains
duplications from $B_{\gx,A,2}(v_i)$ and from some of
$B_{\gx-VE}(e_{i,j})$, $B_{\gx-VE}(e_{i,h})$,
$B_{\gx-VE}(e_{i,k})$, $B_{\gx,A,1}(v_i)$.
Consider a substring $s$ of $B_{\gx-VE}(v_i)$ labeled by
a duplication from a substring $t$ of $B_{\gx,A,2}(v_i)$. First, notice that 
if this duplication is not maximal,
we can extend this duplication as a maximal duplication from a substring $s'$ that includes $s$ to a substring $t'$ 
that includes $t$, without increasing the cost of the labeling.
Notice that then $s'$ is labeled as in a $\tyb$.

Now, we show how to modify $L$ into a labeling $L'$, which is a $\tyb$, without increasing 
the cost of the solution.
$L'$ defines a labeling of $B_{\gx-VE}(v_i)$ by iterating the following procedure.
Denote with $s^*$ be the substring of $B_{\gx-VE}(v_i)$ already labeled by $L'$ in the procedure.
First $s^*= s'$, that is $L'$ labels the string $s'$ as a duplication from $t'$.
Then, consider the unmatched character $\alpha$ of $B_{\gx-VE}(v_i)$ on the left of $s^*$ (if it exists). 
If $\alpha \neq z_1$, 
$L'$ defines a maximal duplication from a substring of $B_{\gx,A,2}(v_i)$ 
to a substring $s''$ on the left of $s^*$ (as in $\tyb$ solution);
if $\alpha = z_1$, $L'$ labels $\alpha$ as a loss.
Similarly, consider the character $\beta$ on the right of $s^*$ (if it exists). 
If $ \beta \neq z_8$, $L'$ defines a maximal duplication from a substring of $B_{\gx,A,2}(v_i)$ 
to a substring of $B_{\gx-VE}(v_i)$ on the right of $s'$ (as in $\tyb$ solution);
if $\beta = z_8$, $L'$ labels $\beta$ as a loss.

Iterating this procedure, we define a labeling $L'$ having the same
cost as $L$, since at each step of the iteration, 
the cost of $L'$ with respect to $L$ is never increased. 
Indeed, consider an unmatched character adjacent to $s^*$, assume w.l.o.g. that this character $\alpha$ is on the left
of $s^*$. $L$ labels $\alpha$ with some label whose cost has not been considered
in previous iterations. At each step the procedure defines a duplication of
maximal length having $\alpha$ as right endpoint.
Indeed, by construction maximal duplications
from $B_{\gx,A,2}(v_i)$ and from
$B_{\gx-VE}(e_{i,j})$, $B_{\gx-VE}(e_{i,h})$,
$B_{\gx-VE}(e_{i,k})$, $B_{\gx,A,1}(v_i)$ have different  start and ending positions
in $B_{\gx-VE}(v_i)$ (except for the rightmost and the leftmost unmatched characters of $B_{\gx-VE}(v_i)$).

Since $L'$ is a $\tyb$ $B_{\gx-VE}(v_i)$ , and $L$ has the same cost of $L'$, 
it follows that $L$ has a cost of at least $8$.
\end{proof}

Now, we prove a property on the labeling of a block 
$B_{\gx-VE}(e_{i,j})$.

\begin{lemma}
\label{lem:sol-loc-edges}
Let $G=(V,E)$ be an instance of MVCC and let $(\gx,\gy)$ be the 
corresponding instance of MLA. Then,
each feasible  alignment of $B_{\gx-VE}(e_{i,j})$ has a
cost of at least $2$; 
furthermore, if an alignment of $B_{\gx-VE}(e_{i,j})$ has a cost of $2$, 
then $B_{\gx-VE}(e_{i,j})$ is labeled
with one duplication from $B_{\gx-VE}(v_i)$
or with one duplication from $B_{\gx-VE}(v_j)$.
\end{lemma}
\begin{proof}
Consider the block $B_{\gx-VE}(e_{i,j})$. 
By construction, since $B_{\gx-VE}(e_{i,j})$ contains $4$ unmatched characters 
and since there is no other substring in 
$\gx$ that is identical to $B_{\gx-VE}(e_{i,j})$, 
it follows that any labeling of $B_{\gx-VE}(e_{i,j})$ requires a cost of at least $2$. 

Now, assume that $B_{\gx-VE}(e_{i,j})$ is not labeled by a duplication 
from $B_{\gx-VE}(v_i)$ or from $B_{\gx-VE}(v_j)$. 
It follows that either each character of $B_{\gx-VE}(e_{i,j})$ is labeled as a loss (hence
the cost of such labeling is exactly $4$)
or the substring $e_{i,j,1} , e_{i,j,2}$ of $B_{\gx-VE}(e_{i,j})$
is labeled as a duplication from $B_{\gx,A,2}(v_i)$. By construction, 
this implies that the leftmost unmatched character of
$B_{\gx-VE}(e_{i,j})$ is either a duplication of length $1$ or a loss, 
and similarly, the rightmost unmatched character of
$B_{\gx-VE}(e_{i,j})$ is either a duplication of length $1$ or a loss. 
Hence this labeling of $B_{\gx-VE}(e_{i,j})$ has a cost of $3$. 
\end{proof}

Now, we are ready to prove the two main properties of the reduction in Lemma \ref{lem:hardness1}
and in Lemma \ref{lem:hardness2}.

\begin{lemma}
\label{lem:hardness1}
Let $G$ be an instance of MVCC and let $(\gx,\gy)$ be the 
corresponding instance of MLA. Then,
given a vertex cover $V' \subseteq V$ of $G$, we can compute in
polynomial time a solution of MLA over instance $(\gx,\gy)$
of cost at most $8 |V'| + 7 |V \setminus V'| + 2 |E|$.
\end{lemma}
\begin{proof}
Let $V'$ be a cover of $G$. We define a solution of MLA over instance
$(\gx,\gy)$ by labeling $\gx$.
First we define the following labeling of block $B_{\gx-VE}(v_i)$, for each $v_i \in V$:
\begin{itemize}
\item for each $v_i \in V'$, define a $\tyb$ for the corresponding 
block $B_{\gx-VE}(v_i)$ (hence the labeling of this block has a cost of $8$, see Remark \ref{Rem:SolV});

\item for each $v_i \in V \setminus V'$, define a $\tya$ for the corresponding 
block $B_{\gx-VE}(v_i)$ (hence the labeling of this block has a cost of $8$, see Remark \ref{Rem:SolV});

\end{itemize}

Now, for each $\{ v_i, v_j \} \in E$ (assume w.l.o.g. $i < j$), 
we define a labeling of the corresponding block $B_{\gx-VE}(e_{i,j})$ as follows:
\begin{itemize}
\item if $v_i \in V'$, define a duplication 
from $B_{\gx-VE}(v_i)$ to $B_{\gx-VE}(e_{i,j})$
(more precisely a duplication for the rightmost three unmatched characters of
$B_{\gx-VE}(e_{i,j})$) and a loss for the leftmost unmatched character of $B_{\gx-VE}(e_{i,j})$;
\item else (notice that in this case $v_j$ must be in $V'$), 
define a duplication from $B_{\gx-VE}(v_j)$ to $B_{\gx-VE}(e_{i,j})$
(more precisely a duplication for the leftmost unmatched characters of
$B_{\gx-VE}(e_{i,j})$) and a loss for the rightmost unmatched character of $B_{\gx-VE}(e_{i,j})$.
\end{itemize}

Notice that, since $V'$ is a vertex cover of $G$, at least one
of $v_i, v_j \in V'$, hence this labeling is always possible.

Now, we show that this labeling is feasible (that is no cycle is induced by the labeling). 
By construction, a block $B_{\gx-VE}(v_i)$ has a duplication coming from a block
$B_{\gx-VE}(e_{i,j})$, only if there is no other block of $\gx$ with a
duplication coming from $B_{\gx-VE}(v_i)$. In case 
a block $B_{\gx-VE}(v_i)$ has a duplication coming from a block
$B_{\gx-VE}(e_{i,j})$, the labeling of $B_{\gx-VE}(e_{i,j})$
defines a duplication from $B_{\gx-VE}(v_j)$ to $B_{\gx-VE}(e_{i,j})$, and
$B_{\gx-VE}(v_j)$ has duplications coming only from
$B_{\gx,A,2}(v_i)$, which does not need any labeling hence it has no incoming arc.
Hence, no cycle is induced by this labeling.
\end{proof}

\begin{lemma}
\label{lem:hardness2}
Let $G$ be an instance of MVCC and let $(\gx,\gy)$ be the 
corresponding instance of MLA. Then,
given a feasible labeling of $(\gx,\gy)$ 
of cost $8 p + 7 (|V| -p) + 2 |E|$, we can compute in
polynomial time a vertex cover of $G$ of size at most $p$.
\end{lemma}
\begin{proof}
Let $L$ be a feasible labeling of $(\gx,\gy)$ 
of cost $8 p + 7 (|V| -p) + 2 |E|$.
First, we consider the labeling of each block
$B_{\gx-VE}(v_i)$, with $v_i \in V$. 
By Lemma \ref{lem:sol-loc-vertex}, we can assume that
$B_{\gx-VE}(v_i)$ is either a $\tya$ or a $\tyb$. Indeed,
if the cost of the labeling of $B_{\gx-VE}(v_i)$ is $7$,
then by Lemma \ref{lem:sol-loc-vertex}, it must be a
$\tya$. If the cost of the labeling of $B_{\gx-VE}(v_i)$ is greater than
$7$, then we can modify (in polynomial time) the labeling of $B_{\gx-VE}(v_i)$
so that it is a $\tyb$ solution. Notice that this modification does not
induce any cycle in $L$, since it defines duplications from $B_{\gx,A,2}(v_i)$
to $B_{\gx-VE}(v_i)$, and $B_{\gx,A,2}(v_i)$ does not need any labeling, hence it 
has no incoming arc.  

Now, consider a block $B_{\gx-VE}(e_{i,j})$, with $\{ v_i, v_j \} \in E$.
We show that we can assume that at least one of $B_{\gx-VE}(v_i)$, $B_{\gx-VE}(v_j)$
has a $\tyb$ in $L$.
Assume to the contrary that both $B_{\gx-VE}(v_i)$, $B_{\gx-VE}(v_j)$
have both a $\tya$. Then by Lemma \ref{lem:sol-loc-edges}, the cost of 
the labeling of $B_{\gx-VE}(e_{i,j})$ has a cost of at least $3$, as 
$B_{\gx-VE}(e_{i,j})$ obviously cannot contain duplications
from $B_{\gx-VE}(v_i)$, $B_{\gx-VE}(v_j)$, otherwise $L$
would induce a cycle and it would not be feasible. 
Now, starting from $L$, we compute in polynomial time a feasible labeling $L'$
such that $c(L') \leq c(L)$, 
as follows: we define a $\tyb$ for one of 
$B_{\gx-VE}(v_i)$, $B_{\gx-VE}(v_j)$, w.l.o.g. $B_{\gx-VE}(v_i)$, 
and we define a duplication from $B_{\gx-VE}(v_i)$ to $B_{\gx-VE}(e_{i,j})$
(for the substring $i-enc_{i,j}$, and a loss for the character $x_{j,q}$, $1 \leq q \leq 3$, 
of $B_{\gx-VE}(e_{i,j})$ not labeled as a duplication from $B_{\gx-VE}(v_i)$. 
Notice that, since $L$ is feasible, the labeling $L'$ is feasible, 
since $B_{\gx-VE}(v_i)$ is a $\tyb$,
hence the duplications of $B_{\gx-VE}(v_i)$ come from 
$B_{\gx,A,2}(v_i)$, that does not have any label and no incoming arc. 
Furthermore, notice that $c(L') \leq c(L)$, 
since we have increased of $1$ the cost of the labeling of 
$B_{\gx-VE}(v_i)$, changing from a $\tya$ to a $\tyb$,
while we have decreased of at least $1$ the cost of labeling $B_{\gx-VE}(e_{i,j})$.

As a consequence we can assume that $L$ is a feasible labeling with the following properties:
(1) each block $B_{\gx-VE}(v_i)$ has either a $\tya$ or a $\tyb$; 
(2) for each block $B_{\gx-VE}(e_{i,j})$, at least one of $B_{\gx-VE}(v_i)$, $B_{\gx-VE}(v_j)$
has a $\tyb$. 
We define a vertex cover $V'$ of $G$ as follows:
\[V' = \{ v_i :  B_{\gx-VE}(v_i) \text{ has a } \tyb\}\]

Since for each $B_{\gx-VE}(e_{i,j})$ at least one of $B_{\gx-VE}(v_i)$, 
$B_{\gx-VE}(v_j)$ has a $\tyb$, it follows that
$V'$ is a vertex cover of $G$. 
Furthermore, since the cost of $L$ is at most $8p+7(|V|-p)+2|E|$, it follows
that $|V'| \leq p$.
\end{proof}

\begin{Thm}
\label{teo:hardness}
MLA is APX-hard.
\end{Thm}
\begin{proof}
The proof follows from Lemma \ref{lem:hardness1} and from Lemma \ref{lem:hardness2},
and from the observation that in a cubic graph $|E|= \frac{3}{2}|V|$
and a vertex cover has size at least $\frac{|V|}{4}$.
\end{proof}

\end{document}